\documentclass[11pt]{article}
\usepackage{fullpage, amsfonts,amssymb,amsmath,amsthm}

\usepackage{times}

\usepackage{color}

\newtheorem{puzzle}{Puzzle}

\newcommand{\seq}{\{0, 1\}^{\omega}}

\DeclareMathOperator{\DP}{DP}

\DeclareMathOperator{\success}{Suc}

\DeclareMathOperator{\poly}{poly}

\newtheorem{theorem}{Theorem}

\newtheorem{claim}{Claim}

\newtheorem{lemma}{Lemma}

\renewcommand{\mathbf}[1]{{\pmb #1}}

\newcommand{\eps}{\varepsilon}

\newcommand{\cS}{{\cal S}}
\newcommand{\cT}{{\cal T}}

\newcommand{\bE}{\mathbb{E}}

\newcommand{\bN}{\mathbb{N}}

\begin{document}

\title{High-Confidence Predictions under Adversarial Uncertainty}

\author{Andrew Drucker\thanks{MIT, EECS Dept., Cambridge, MA, USA.\  Email: adrucker@mit.edu.\ Supported by a DARPA YFA grant.}}

\date{}

\maketitle

\begin{abstract}
We study the setting in which the bits of an unknown infinite binary sequence $x$ are revealed sequentially to an observer.  We show that very limited assumptions about $x$ allow one to make successful predictions about unseen bits of $x$.  First, we study the problem of successfully predicting a single 0 from among the bits of $x$.  In our model we have only one chance to make a prediction, but may do so at a time of our choosing.  We describe and motivate this as the problem of a frog who wants to cross a road safely.

Letting $N_t$ denote the number of 1s among the first $t$ bits of $x$, we say that $x$ is ``$\eps$-weakly sparse'' if $\lim \inf (N_t/t) \leq \eps$.  Our main result is a randomized algorithm that, given any $\eps$-weakly sparse sequence $x$, predicts a 0 of $x$ with success probability as close as desired to $1 - \eps$.  Thus we can perform this task with essentially the same success probability as under the much stronger assumption that each bit of $x$ takes the value 1 independently with probability $\eps$.
We apply this result to show how to successfully predict a bit (0 or 1) under a broad class of possible assumptions on the sequence $x$.  The assumptions are stated in terms of the behavior of a finite automaton $M$ reading the bits of $x$.

We also propose and solve a variant of the well-studied ``ignorant forecasting'' problem.  For every $\eps > 0$, we give a randomized forecasting algorithm $\mathcal{S}_{\eps}$ that, given sequential access to a binary sequence $x$, makes a prediction of the form: ``A $p$ fraction of the next $N$ bits will be 1s.''  (The algorithm gets to choose $p, N$, and the time of the prediction.)   For any fixed sequence $x$, the forecast fraction $p$ is accurate to within $\pm \eps$ with probability $1 - \eps$.

\end{abstract}


\section{Introduction}\label{introsec}


\subsection{The frog crossing problem}

A frog wants to cross the road at some fixed location, to get to a nice pond.  But she is concerned about cars.  It takes her a minute to cross the road, and if a car passes during that time, she will be squashed.  However, this is no ordinary frog.  She is extremely patient, and happy to wait any finite number of steps to cross the road.  What's more, she can observe and remember how many cars have passed, as well as when they passed.  She can follow any algorithm to determine when to cross the road based on what she has seen so far, although her senses aren't keen enough to detect a car before it arrives.

 Think of a ``car-stream'' as defined by an infinite sequence of 0s and 1s describing the minutes when a car passes (we model time as discrete, and assume that at most one car passes each minute).
We ask, under what assumptions on the car-stream can our frog cross the road safely?  Obviously, if there is constant, bumper-to-bumper traffic---the car-stream described by the sequence $(1, 1, 1, \ldots)$---then she cannot succeed, so we must make some assumption.

One natural approach to this kind of situation is to assume traffic is generated according to some \emph{probabilistic model}.  For example, we might assume that during each minute, a car arrives with probability $.1$, and that these events are independent.  More complicated assumptions (involving dependence between the car arrival-times, for example) can also be considered.

However, the frog may not have a detailed idea of how the cars are generated.  It may be that the frog merely knows or conjectures some \emph{constraint} obeyed by the car-stream.
We then ask whether there exists a strategy which gets the frog safely across the road (at least, with sufficiently high probability), for \emph{any} car-stream obeying the constraint.   This model will be our focus in the present paper.
For example, suppose the cars appear to arrive more-or-less independently with probability $.1$ at each minute.  The frog may be unsure that the independence assumption is fully justified, so she may make the weaker assumption that the limiting car-density is at most $.1$.

Note that this condition holds with probability 1 if the cars really are generated by independent $.1$-biased trials, so this constraint can be considered a natural \emph{relaxation} of the original probabilistic model.
The frog may then ask whether there exists a strategy that gets her across the road safely with probability nearly $.9$ under this relaxed assumption.  (Happily, the answer is Yes; this will follow from our main result.)

\subsection{Relation to previous work}\label{relatedsec}

Our work studies \emph{prediction under adversarial uncertainty}.  In such problems, an observer tries to make predictions about successive states of nature, without assuming that these states are governed by some known probability distribution.  Instead, nature is regarded as an \emph{adversary} who makes choices in an attempt to thwart the observer's prediction strategy.
 The focus is on understanding what kinds of predictions can be made under very limited assumptions about the behavior of nature.  

Adversarial prediction is a broad topic, but two strands of research are particularly related to our work.  The first strand is the study of \emph{gales} and their relatives.  Gales are a class of betting systems generalizing martingales; their study is fundamental for the theory of \emph{effective dimension} in theoretical computer science (see~\cite{fgcc} for a survey).  The basic idea is as follows.  An infinite sequence $x$ is chosen from some known subset $A$ of the space $\{0, 1\}^{\omega}$ of infinite binary sequences.  A gambler is invited to gamble on predicting the bits of $x$ as they are sequentially revealed; the gambler has a finite initial fortune and cannot go into debt.  The basic question is, for which subsets $A$ can the gambler be guaranteed long-term success in gambling, for any choice of $x \in A$?  This question can be studied under different meanings of ``success'' for the gambler, and under more- or less-favorable classes of bets offered by the casino.  

Intuitively, the difficulty of gambling successfully on an unknown $x \in A$ is a measure of the ``largeness'' of the set $A$.  In fact, this perspective was shown to yield new characterizations of two important measures of fractal dimensionality.  Lutz~\cite{lutz03} gave a characterization of the Hausdorff dimension of subsets of $\{0, 1\}^{\omega}$ in terms of gales, while Athreya, Hitchcock, Lutz, and Mayordomo~\cite{ath07} showed a gale characterization of the packing dimension.  These works also investigated gales with a requirement that the gambler follows a computationally bounded betting strategy; using such gales, the authors explored new notions of ``effective dimension'' for complexity classes in computational complexity theory.\footnote{Computationally bounded betting and prediction schemes have also been used to study \emph{individual} sequences $x$, rather than sets of sequences.  This approach has been followed using various resource bounds and measures of predictive success; see, e.g.,~\cite{mer98, lutz_indiv}.}

The second strand of related work is the so-called \emph{forecasting problem} in decision theory (see \cite{dawid} for an early, influential discussion).  In this problem, an infinite binary sequence $x \in \{0,1\}^\omega$ is once again revealed sequentially; we typically think of the $t$-th bit as indicating whether it rained on the $t$-th day at some location of interest.  Each day a weather forecaster is asked to give, not an absolute prediction of whether it will rain tomorrow, but instead some estimate of the \emph{probability} of rain tomorrow.  In order to keep his job as the local weather reporter, the forecaster is expected to make forecasts which have the property of being \emph{calibrated}: roughly speaking, this means that if we consider all the days for which the forecaster predicted some probability $p$ of rain, about a $p$ fraction turn out rainy (see~\cite{fostv98} for more precise definitions).

In the adversarial setting, a forecaster must make such forecasts without knowledge of the probability distribution governing nature.  In the well-studied ``ignorant forecaster'' model, the forecaster is allowed \emph{no assumptions whatsoever} about the sequence $x$.  Nevertheless, it is a remarkable fact, shown by Foster and Vohra~\cite{fostv98}, that there exists a randomized ignorant forecasting scheme whose forecasts are calibrated in the limit.  

This result was extended by Sandroni~\cite{sand03}.  The calibration criterion is just one of many conceivable ``tests'' with which we might judge a forecaster's knowledge on the basis of his forecasts and the observed outcomes.  Foster and Vohra's result showed that the calibration test can be passed even by an ignorant forecaster; but conceivably some other test of knowledge could be more meaningful.  A reasonable class of tests to consider are those that can be passed with some high probability $1 - \eps$ by a forecaster who knows the actual distribution $\mathcal{D}$ governing nature, for any possible setting of $\mathcal{D}$.
However, Sandroni showed that any such test can \emph{also} be passed with probability $1 - \eps$ by an ignorant forecaster!  Fortnow and Vohra~\cite{fortv09} give evidence that the ignorant strategies provided by Sandroni's result cannot in general be computed in polynomial time, even if the test is polynomial-time computable.\footnote{The tests considered in~\cite{sand03, fortv09} are required to halt with an answer in finite time.  See~\cite{fortv09} for references to work in which this restriction is relaxed.}

In both of the strands of research described above, researchers have typically looked for prediction schemes that have some desirable \emph{long-term, aggregate} property.  In the gale setting, the focus is on betting strategies that may lose money on certain bets, but that succeed in the limit; in the forecasting problem, an ignorant forecaster wants his forecasts to appear competent overall, but is not required to give definite predictions of whether or not it will rain on any given day.
By contrast, in our frog problem, the frog wants to cross the road just once, and her life depends on the outcome.  Our focus is on making a \emph{single prediction}, with success probability as close to 1 as possible.  

In a later section of the paper we will also study a variant of the ignorant forecasting scenario.  Following~\cite{fostv98, sand03}, we will make no assumption about the observation sequence $x$.  Our goal will be to make a single forecast at a time of our choosing, of the following form: ``A $p$ fraction of the next $N$ observations will take the value 1.''  We will seek to maximize the accuracy of our prediction, as well as the likelihood of falling within the desired accuracy.  This forecasting variant is conceptually linked to our frog problem by its focus on making a single prediction with high confidence.

\subsection{Our results on the frog crossing problem}\label{puzsec}

We now return to our patient frog.

To appreciate the kinds of frog-strategies that are possible, we first consider a simple but instructive example.  Suppose the frog knows that at most \emph{one} car will ever drive by.  In this case, the frog might choose to wait until she sees a car pass; however, this strategy makes her wait forever if no car ever arrives, and we consider this a failure.  Similarly, suppose the frog follows a deterministic strategy which, for some $t \geq 1$, makes her cross on the $t$-th minute if she has not yet seen a car.  Then the frog is squashed on the car-stream consisting of a single car passing at the $t$-th minute.  Thus, any deterministic strategy fails against some car-stream obeying our constraint.

What is left for the frog?   We recommend following a \emph{randomized} strategy.  Fixing  some $\delta > 0$, consider the following strategy: the frog  chooses a value $t^\star \in \{1, 2, \ldots, \lceil 1/\delta \rceil\}$ uniformly at random, and crosses at time $t^\star$.  Let's analyze this algorithm.  Fix any car-stream consisting of at most one car, say arriving at time $t \geq 1$ (where $t := \infty$ if no car arrives).  Then the strategy above fails only if $t^\star = t$, which occurs with probability at most $ \lceil 1/\delta \rceil^{-1} \leq \delta$.

Note that this error probability is over the randomness in the \emph{algorithm}, not the car-stream; we regard the car-stream as chosen by an adversary who knows the frog's strategy, but \emph{not} the outcomes of the frog's random decisions.  We are interested in strategies which succeed with high probability against any choice by the adversary (obeying the assumed constraint).

An easy modification of the above algorithm lets the frog succeed with probability $1 - \delta$ against a car-stream promised to contain at most $M$ cars, for any \emph{fixed} $M < \infty$.  However, it may come as a surprise that we can succeed given a much weaker assumption.  The reader is invited to try the following puzzle:

\begin{puzzle}\label{puz1} For any $\delta > 0$, give a frog-strategy that succeeds with probability $1 - \delta$, under the assumption that the number of cars is finite.
\end{puzzle}

The assumption can be weakened further.  Fixing a car-stream, let $N_t$ denote the number of cars appearing in the first $t$ minutes.  Say that the car-stream is \emph{sparse} if $N_t = o(t)$, that is, if
\[\lim_{t \rightarrow \infty }  N_t/t   = 0.     \]
\begin{puzzle}\label{puz2} Give a frog-strategy that succeeds with probability $1 - \delta$, under the assumption that the car-stream is sparse.
\end{puzzle}

Note that in Puzzle~\ref{puz2}, the frog is promised that the fraction $N_t/t$ approaches 0 as $t \rightarrow \infty$, but she has no idea how quickly it will do so.

Say the car-stream is \emph{weakly sparse} if $N_t \notin \Omega(t)$, that is, if
\[\lim_{s \rightarrow \infty }  \left( \inf_{t \geq s}  N_t/t  \right)  = 0.     \]

\begin{puzzle}\label{puz3} Give a frog-strategy that succeeds with probability $1 - \delta$, under the assumption that the car-stream is weakly sparse.
\end{puzzle}

In this paper we provide a solution to Puzzle~\ref{puz3}.  This immediately implies a solution for Puzzles~\ref{puz1} and~\ref{puz2}, but these first two puzzles also have simpler solutions, which we encourage the reader to find.  
The basic idea of our solution to Puzzle~\ref{puz3} is easy to state: roughly speaking, seeing fewer cars increases the frog's ``courage'' and makes her more likely to decide to cross.  Correctly implementing and analyzing this idea turns out to be a delicate task, however.

We actually prove a quantitative strengthening of Puzzle~\ref{puz3}.  For any $\eps> 0$, say that a car-stream is \emph{$\eps$-weakly sparse} if
\[\lim_{s \rightarrow \infty }  \left( \inf_{t \geq s}  N_t/t  \right)  \leq \eps.     \]
Our main result is that, under the assumption that the car-stream is $\eps$-weakly sparse, the frog can cross successfully with probability as close as desired to $1 - \eps$.   We state our result formally in Section~\ref{prelimsec} after setting up the necessary definitions.

Our result bears some resemblance to known results in dimension theory.   Let $A_{\eps-ws} \subseteq \seq$ denote the set of $\eps$-weakly sparse infinite binary sequences.  Eggleston~\cite{egg, ergo} showed that for $\eps \leq 1/2$, the Hausdorff dimension of $A_{\eps-ws}$ is equal to the binary entropy $H(\eps)$.  More recently, Lutz~\cite{lutz03} gave an alternative proof using his gale characterization of Hausdorff dimension (Lutz also calculated the ``effective dimension'' of $A_{\eps-ws}$ according to several definitions).  Lutz upper-bounds the Hausdorff dimension of $A_{\eps-ws}$ by giving a gale betting strategy that ``succeeds'' (in the appropriate sense) against all $x \in A_{\eps-ws}$.  This betting strategy, which is simple and elegant, does not appear to be applicable to our puzzles.  Indeed, a major difference between our work and the study of gales is that gale betting strategies are \emph{deterministic} (at least under standard definitions~\cite{lutz03, ath07}), whereas randomization plays a crucial role in our frog-strategies.

\subsection{Further results}\label{sec:further}

 In Section~\ref{automsec}, we prove an extension of the result of Puzzle~\ref{puz3}, in a a modified setting in which we are allowed to predict either a 0 or a 1.  We give a condition on the binary sequence $x$ that is significantly more general than weak sparsity, and that still allows a bit to be predicted with high confidence.   The condition is stated in terms of a finite automaton $M$ that reads $x$: we assume that $x$ causes $M$ to enter a designated set of ``bad'' states $B$ only infrequently.  A certain ``strong accessibility'' assumption on the states $B$ is needed for our result.
 
 In Section~\ref{forecastsec}, we study a problem closely related to the ``ignorant forecasting'' problem discussed earlier, where (as in the frog problem) a single prediction is to be made.  In the ``density prediction game,'' an arbitrary infinite binary sequence is chosen by Nature, and its bits are revealed to us sequentially.  Our goal is to make a single forecast of the form
\[ \text{\textit{``A $p$ fraction of the next $N$ bits will be 1s.''}}  \]
We are allowed to choose $p, N$, and the time at which we make our forecast.

Fixing a binary sequence $x$, we say that a forecast described by $(p, N)$, and made after viewing $x_t$, is \emph{$\eps$-successful} on $x$ if the fraction of 1s among $x_{t+1}, \ldots, x_{t + N}$ is in the range $(p - \eps, p + \eps)$.
For $\delta, \eps > 0$, we say that a (randomized) forecasting strategy $\mathcal{S}$ is \emph{$(\delta, \eps)$-successful} if for every $x \in \seq$,
\[ \Pr[\mathcal{S}\text{ is }\eps\text{-successful on $x$} ] \geq 1 - \delta.\]

In Section~\ref{forecastsec}, we show the following, perhaps surprising, result:

\begin{theorem}\label{thm:denspred} For any $\delta, \eps > 0$, there exists a $(\delta, \eps)$-successful forecasting strategy.
\end{theorem}

\section{Preliminaries and the Main Theorem}\label{prelimsec}

First we develop a formal basis to state and prove our main result.  $\mathbb{N} = \{1, 2, \ldots\}$ denotes the positive whole numbers.  For $N \in \bN$, $[N]$ denotes the set $\{1, 2, \ldots, N\}$.  $\{0, 1\}^{\omega}$ denotes the set of all infinite bit-sequences $b = (b_1, b_2, \ldots )$.  We will freely refer to any such sequence as a ``car-stream,'' where $b_i = 1$ means ``a car appears during the $i$-th minute.''

A \emph{frog-strategy} (or simply \emph{strategy}) is a collection 
\[ \mathcal{S} = \{\pi_{\cS, b} : b \in \seq  \},\]
where each $\pi_{\cS, b}$ is a probability distribution over $\bN \cup \{\infty\}$.  We require that for all $b = (b_1, b_2, \ldots), b' = (b'_1, b'_2, \ldots)$, and all $i \in \bN$,
\begin{equation}\label{eq:restrict}
(b_1, \ldots, b_{i-1} ) =  (b'_1, \ldots, b'_{i-1} )   \Rightarrow  \pi_{\cS, b}(i) = \pi_{\cS, b'}(i).
\end{equation}
That is, $\pi_{\cS, b}(i)$ depends only on $b_1, \ldots, b_{i - 1}$.  
 
Let us interpret the above definition.  A frog-strategy defines, for each car-stream $b$ and each $i \in \bN$, a probability $\pi_{\cS, b}(i)$ that, when facing the car-stream $b$, the frog will attempt to cross at step $i$.  There is also some probability $\pi_{\cS, b}(\infty)$ that the frog will wait forever without crossing.  Whether it lives or dies, the frog only attempts to cross at most once, so these probabilities sum to 1.  Eq.~(\ref{eq:restrict}) requires that the frog's decision for the $i$-th minute depends only upon what it has seen of the car-stream during the first $(i - 1)$ minutes.  The frog-strategies we analyze in this paper will be defined in such a way that Eq.~(\ref{eq:restrict}) obviously holds.

Given a frog-strategy $\mathcal{S}$, define the success probability
\[   \success(\mathcal{S}, b)  :=  \sum_{i \in \bN: b_i = 0} \pi_{\cS, b}(i)         \]
as the probability that, facing $b$, the frog crosses successfully during some minute when there is no car.  
Similarly, define the \emph{death probability} 
\[\DP(\cS, b) := \sum_{i \in \bN: b_i = 1} \pi_{\cS, b}(i)   = 1 - \success(\cS, b) -   \pi_{\cS, b}(\infty)  \]
as the probability that the strategy $\cS$ leads to the frog being squashed by a car on car-stream $b$.
For a subset $A \subseteq \seq$, define 
\[\success(\mathcal{S}, A) := \inf_{b \in A} \success(\mathcal{S}, b).\]

Let $N_t = N_t(b) := b_1 + \ldots + b_t$.  A car-stream $b$ is called \emph{$\eps$-weakly sparse} if
\[\lim_{s \rightarrow \infty }  \left( \inf_{t \geq s}  N_t/t  \right)  \leq \eps.     \]
We can now formally state our main result:

\begin{theorem}\label{thm:main}  Fix $\eps \in (0, 1)$ and let $A_{\eps-ws} := \{b: b$ is $\eps-$weakly sparse$\}$.  Then for all $\gamma > 0$, there exists a strategy $\mathcal{S}_{\eps, \gamma}$ such that
\[    \success(\mathcal{S}_{\eps, \gamma}, A_{\eps-ws})    > 1 - \eps - \gamma.         \]
Furthermore, $\mathcal{S}_{\eps, \gamma}$ has the following ``safety'' property: for \emph{any} car-stream $b \in \{0, 1\}^\omega$, the death probability $\DP(\cS, b)$ is at most $\eps + \gamma$.
\end{theorem}

Note that $b$ is weakly sparse (as defined in Section~\ref{puzsec}) exactly if it is $\eps$-weakly sparse for all $\eps > 0$.  Thus if $A_{ws} := \{b: b$ is weakly sparse$\}$, then by Theorem~\ref{thm:main}, we can succeed on $A_{ws}$ with probability as close to 1 as we desire.  This solves Puzzle 3.

It is not hard to see that Theorem~\ref{thm:main} is optimal for frog-strategies against $A_{\eps-ws}$.  For consider a randomly generated car-stream $\mathbf{b}$ where the events $[\mathbf{b}_i = 1]$ occur independently, with $\bE[\mathbf{b}_i] = \min \{1, \eps + 2^{-i}\}$.  Then $[\lim_{t \rightarrow \infty} N_t(\mathbf{b})/t = \eps]$ occurs with probability 1.  On the other hand, any frog-strategy $\cS$ has success probability less than $1 - \eps$ against $\mathbf{b}$.  Thus, for any $\cS$ we can find a particular car-stream $b$ for which $\lim_{t \rightarrow \infty} N_t(b)/t = \eps$ and which causes $\cS$ to succeed with probability less than $1 - \eps$.


\section{Proof of the the Main Theorem}\label{mainsec}   

In this section we prove Theorem~\ref{thm:main}.  First we observe that, if we can construct a strategy $\mathcal{S}$ such that $\success(\mathcal{S}, A_{\eps-ws})    > 1 - \eps - \gamma$, then the ``safety'' property claimed for $\mathcal{S}$ in the theorem statement will follow immediately.  For suppose to the contrary that some car-stream $b \in \seq$ satisfies $\DP(\cS, b) > \eps + \gamma$.  Then there exists $m \in \mathbb{N}$ such that $\sum_{i \leq m: b_i = 1} \pi_{\cS, b}(i) > \eps + \gamma$.  If we define $b' \in \seq$ by
\[  b'_i :=  
\left\{
	\begin{array}{ll}
		 b_i & \mbox{if } i \leq m, \\
		0 & \mbox{if } i > m,
	\end{array}
\right. \]
then $b' \in A_{\eps-ws}$ and $\DP(\cS, b') > \eps + \gamma$, contradicting our assumption on $\mathcal{S}$.

To construct the strategy $\mathcal{S}$, we use a family of frog-strategies for attempting to cross the road within a finite, bounded interval of time.  The following lemma is our key tool, and is interesting in its own right.

\begin{lemma}\label{lem:finite} For any $\delta \in (0, 1)$ and integer $K > 1$, there exists a strategy $\cT = \cT_{K, \delta}$ such that for all $b \in \seq$:
\begin{itemize} 

\item[(i)] The crossing time of $\cT$ is always in $[K] \cup \{\infty\}$.  That is, for $K < i < \infty$, we have $\pi_{\cT, b}(i) = 0$;

\item[(ii)] If $(b_1 + \ldots + b_{K-1})/(K-1)  \leq \delta' < \delta$, then $\pi_{\cT, b}(\infty)  \leq 1 - \Omega((\delta - \delta')^2/\delta)$;  

\item[(iii)] The death probability satisfies
\[\DP(\cT, b) \leq \frac{\delta}{1 - \delta} \success(\cT, b) + O\left( \frac{\delta}{(1 - \delta)K} \right).\]
\end{itemize} 
\end{lemma}

We defer the proof of Lemma~\ref{lem:finite}, and use it to prove Theorem~\ref{thm:main}.

\begin{proof}[Proof of Theorem~\ref{thm:main}]  Fix settings of $\eps, \gamma > 0$; we may assume $\eps + \gamma < 1$, or there is nothing to prove.  Let $\eps_1 := \eps + \gamma/3, \eps_2 := \eps + 2\gamma/3$.  
We also use a large integer $K > 1$, to be specified later.
Divide $\bN$ into a sequence of intervals $I_1 = \{1, 2, \ldots, K\}, I_2 = \{K+1, \ldots, 5K\}$, and so on, where $I_r$ has length $r^2 K$.

Let $\cS = \cS_{\eps, \gamma}$ be the frog-strategy which does the following: first, follow the strategy $\cT_{K, \eps_2}$ (as given by Lemma~\ref{lem:finite}) during the time interval $I_1$.  If no crossing is attempted during these steps, then run the strategy $\cT_{4K, \eps_2}$ on the interval $I_2$, after shifting the indices of $I_2$ appropriately (so that $\cT_{4K, \eps_2}$ considers its input sequence to begin on $b_{K+1}$).  
Similarly, for each $r > 0$, if we reach the interval $I_r$ without an attempted crossing, we execute the strategy $\cT_{r^2K, \eps_2}$ on the interval $I_r$, after shifting indices appropriately.

We will show that if $K$ is sufficiently large, we have $\success(\cS, A_{\eps-ws})   >   1 - \eps - \gamma$ as required.  
Fix any $b = (b_1, b_2, \ldots) \in A_{\eps-ws}$.
Let $\alpha_r  := \left(\sum_{i \in I_r}b_i\right)/|I_r| $ be the fraction of 1-entries in $b$ during interval $I_r$.

\begin{claim}\label{ezclaim} For infinitely many $r$, $\alpha_r \leq \eps_1$. 
\end{claim} 
\begin{proof} Suppose to the contrary that $\alpha_r > \eps_1$ when $r \geq R$.  Consider an interval $\{1, 2, \ldots, M\}$ large enough to properly contain $I_1, I_2, \ldots, I_R$.  Let $t \geq R$ be such that $I_t \subseteq [M]$ but that $I_{t+1} \nsubseteq [M]$.  Let $\alpha^{\star}$ be the fraction of 1-entries in $[M] \cap I_{t+1}$; we set $\alpha^\star := 0$ if $[M] \cap I_{t+1} = \emptyset$.  With $N_M = (b_1 + \ldots + b_M)$, we have the expression
\[\frac{N_M}{M}  = \sum_{r \leq t}  \frac{|I_{r}|}{M}  \cdot  \alpha_{r}    +  \frac{|[M] \cap I_{t+1}|}{M}    \cdot \alpha^\star \]
giving the car-density (fraction of 1s) of $b$ in $[M]$ as a weighted average of the car-densities in $I_1, \ldots, I_t$ and in $[M] \cap I_{t+1}$.  

Note that  
\[\displaystyle\frac{|[M] \cap I_{t+1}|}{M}  \leq \frac{|I_{t+1}|}{M}   \leq  \frac{(t + 1)^2K}{\sum_{r \leq t} r^2 K} = O(1/t) \rightarrow 0,\]
as $M \rightarrow \infty$.  Now $\alpha_r > \eps_1$ when $r \geq R$, so for sufficiently large $M$ we have $N_M/M \geq (\eps_1 + \eps)/2 > \eps$.  But this contradicts the fact that $b \in A_{\eps-ws}$, proving the Claim.
\end{proof}

Fix $r > 0$.  If $I_r = \{j, \ldots, k\}$ and $\alpha_r = (b_j + \ldots + b_k)/(k - j + 1) \leq \eps_1$, then we also have $(b_j + \ldots + b_{k - 1})/(k - j ) < \eps + \gamma/2$ if $r$ is large enough.  For any such $r$, condition (ii) of Lemma~\ref{lem:finite} tells us that if our frog-strategy reaches the interval $I_r$, it will attempt to cross during $I_r$ with probability $\Omega((\gamma/6)^2/\eps_2)$.  There are infinitely many such $r$, by Claim~\ref{ezclaim}.  Thus, the frog-strategy eventually attempts to cross with probability 1.  It follows that $\DP(\cS, b)  = 1 -\success(\cS, b)$.

For $r > 0$, let $P_r = P_r(b)$ be defined as the probability that $\cS$ reaches $I_r$ without attempting to cross earlier.  Let $b[I_r]$ denote the sequence $b$, shifted to begin at the first bit of $I_r$.  Then we can reexpress the death probability of $\cS$ on $b$, and bound this quantity, as follows:
{\allowdisplaybreaks \begin{align*}
\DP(\cS, b) &= \sum_{r \geq 1} P_r \cdot \DP(\cT_{r^2K, \eps_2}, b[I_r]) \\
&\leq  \sum_{r \geq 1} P_r \cdot \left(  \frac{\eps_2}{1 - \eps_2} \success(\cT_{r^2K, \eps_2}, b[I_r]) + O\left( \frac{\eps_2}{(1 - \eps_2)r^2 K} \right)   \right)\\
&\text{(by condition (iii) of Lemma~\ref{lem:finite})}\\
&= \frac{\eps_2}{1 - \eps_2}\left(\sum_{r \geq 1} P_r \cdot \success(\cT_{r^2K, \eps_2}, b[I_r]) \right)  + O\left( \frac{ \eps_2}{(1 - \eps_2)K} \right)\\
&\text{(using the fact that $\sum_{r > 0} r^{-2} < \infty$)}\\
&= \frac{\eps_2}{1 - \eps_2} \success(\cS, b) + O\left( \frac{\eps_2}{(1 - \eps_2)K} \right).
\end{align*}    }
Thus, $\DP(\cS, b) =  1 -\success(\cS, b) \leq \frac{\eps_2}{1 - \eps_2} \success(\cS, b) + O\left( \frac{\eps_2}{(1 - \eps_2)K} \right)$, which implies
\[   \success(\cS, b)  \geq 1 - \eps_2 - O\left( \frac{\eps_2}{(1 - \eps_2)K} \right) = 1 - (\eps + 2\gamma/3) - O\left( \frac{\eps_2}{(1 - \eps_2)K} \right).         \]
By setting $K \gg \eps_2 \gamma^{-1}(1 - \eps_2)^{-1}$ sufficiently large, we can conclude $\success(\cS, b)  > 1 - \eps - \gamma$, where the slack in the inequality is independent of the choice of $b \in A_{\eps-ws}$.  This proves Theorem~\ref{thm:main}.
\end{proof}

\begin{proof}[Proof of Lemma~\ref{lem:finite}]  
By an easy approximation argument, it suffices to prove the result for the case when $\delta$ is rational.  So assume 
\[\delta = p/d,\]
for some integers $0 < p < d$, and let
\[ q := d - p.\]
The frog-strategy $\cT$ is as follows.  First, pick a value $t^\star \in [K]$ uniformly at random.  Do not attempt to cross on steps $1, 2, \ldots, t^\star - 1$.  During this time, maintain an ordered stack of ``chips,'' initially empty.  For $1 \leq i < t^\star$, after viewing $b_i$, if $b_i = 0$ then add $p$ chips to the top of the stack; if $b_i = 1$ then remove $q$ chips from the top of the stack---or, if the stack contains fewer than $q$ chips, remove all the chips.  After this modification to the stack, we say that the bit $b_i$ has been ``processed''.

For $0 \leq i \leq K-1$, let $H_i$ denote the number of chips on the stack after processing $b_1, \ldots, b_i$ (so, $H_0 = 0$).
After processing $b_{t^\star - 1}$, sample from a 0/1-valued random variable $X$, with expectation
\[ \bE[X] :=  \frac{H_{t^\star-1}}{d K} .   \]
(Note that this expectation is at most $\frac{p(K - 1)}{dK} < 1$, so the definition makes sense.)
Attempt to cross at step $t^\star$ if $X = 1$, otherwise make no crossing attempt at any step.  

Note that the variable $H_t$ can be regarded as a measure of the frog's ``courage'' after processing $b_1, \ldots, b_t$, as in our sketch-description in Section~\ref{puzsec}.
We now verify that $\cT$ has the desired properties. 
Condition (i) in Lemma~\ref{lem:finite} is clearly satisfied.
Before verifying conditions (ii) and (iii), we first sketch why they hold.  For (ii), the idea is that if much less than a $\delta$ fraction of $b_1, \ldots, b_{K-1}$ are 1s, then the stack of chips will be of significant height after processing these bits.  Since the stack doesn't grow too quickly, we conclude that the \emph{average} stack height during these steps is significant, which implies that the frog attempts to cross with noticeable probability.

For (iii), the idea is that for any chip $c$, if $c$ stays on the stack for a significant amount of time, then the fraction of 1s appearing during the interval in which $c$ was on the stack must be not much more than $\delta$.  Thus $c$'s contribution to the death probability is not much more than $
\delta/(1 - \delta)$ times $c$'s contribution to the success probability.  On the other hand, chips $c$ which don't stay on the stack very long make only a small contribution to the death probability.

Now we formally verify condition (ii).  Fix some sequence $b$.  First note that the placement and removal of chips, and the height sequence $H_0, \ldots, H_{K - 1}$, can be defined in terms of $b$ alone, without reference to the algorithm's random choices.  Throughout our analysis we consider the stack to continue to evolve as a function of the bits $b_1, \ldots, b_{K-1}$, regardless of the algorithm's choices.

Suppose $b_1 + \ldots + b_{K-1} \leq \delta'(K-1)$, where $\delta' < \delta$; we ask, how large can  $\pi_{\cT, b}(\infty)$ be?  From the definition of $\cT$, we compute
\begin{equation}\label{eq:commit}  \pi_{\cT, b}(\infty) =    1 - \frac{1}{K}  \sum_{t \in [K]}   \frac{H_{t-1}}{d K} = 1 - \frac{1}{dK^2}\sum_{0 \leq t < K}H_t. 
\end{equation}
Now, for a chip $c$, let $m_c \in \bN$ denote the number of indices $i < K$ for which $c$ was on the stack immediately after processing $b_i$.  (We consider each chip to be ``unique;'' that is, it is added to the stack at most once.)  We can reexpress the sum appearing in Eq.~(\ref{eq:commit}) as 
\[  \sum_{0 \leq t < K}H_t  = \sum_c  m_c. \]
We will lower-bound this sum by considering the contribution made by chips that are never removed from the stack---that is, chips which remain after processing $b_{K-1}$.  We call such chips ``persistent.''  First, we argue that there are many persistent chips.  By our assumption, at least $p \cdot (1 - \delta')(K - 1)$ chips are added to the stack in total, while at most $q\cdot \delta' (K - 1)$ chips are ever removed.  Thus the number of persistent chips is at least
\begin{align*}p (1 - \delta')(K -1)    -   q\delta' (K - 1)  &=   p(1 - \delta + (\delta - \delta'))(K - 1)   - q(\delta + (\delta' - \delta))(K - 1) \\
&= [\underbrace{p(1 - \delta) - q\delta}_{= 0}   +   \underbrace{(p + q)}_{= d}(\delta - \delta') ](K - 1) \\
&= (\delta - \delta')d (K-1) ,
\end{align*}
where we used $p/q = \delta/(1 - \delta)$.  Let $J := (\delta - \delta')d (K-1)$.

Pick any $J$ persistent chips, and number them $c(1), \ldots, c(J)$ so that $j' < j \leq J$ implies $c(j')$ appears above $c(j)$ on the stack after processing $b_{K-1}$.  This means $c(j')$ was added to the stack no earlier than $c(j)$, so that $m_{c(j')} \leq m_{c(j)}$.
At most $p$ chips are added for every processed bit of $b$, and if $c(j)$ was added while processing the $(K - i)$-th bit, then $m_{c(j)} = i$.  Thus, by our indexing we conclude $m_{c(j)} \geq   \lceil j/p \rceil \geq j/p$.  Summing over $j$, we obtain
\begin{align*}\sum_{\text{persistent }c}m_c &\geq  \sum_{j = 1}^{J}   j/p   \\
&= \frac{J (J+1)}{2p}\\
&>  \frac{(\delta - \delta')^2d^2 (K-1)^2}{2p}\\\
&= \frac{(\delta - \delta')^2d(K - 1)^2}{2\delta}.
\end{align*}
Finally, returning to Eq.~(\ref{eq:commit}), we compute
\[  \pi_{\cT, b}(\infty)   = 1  - \frac{1}{dK^2} \sum_c  m_c   < 1 - \frac{1}{dK^2} \cdot \frac{(\delta - \delta')^2d(K - 1)^2}{2\delta }  < 1 - \frac{(\delta - \delta')^2}{8\delta},\] 
since $K > 1$.  This establishes condition (ii).

Now we verify condition (iii).  Fix any car-stream $b$. 
From our definitions, we have the expressions
\[      \success(\mathcal{S}, b)  =  \frac{1}{K}\sum_{t \in [K]: b_t = 0}\frac{ H_{t - 1}}{dK}  ,  \quad{} \DP(\mathcal{S}, b)  =  \frac{1}{K}\sum_{t \in [K]: b_t= 1}\frac{ H_{t - 1}}{dK},  \quad{ } \text{and so}    \]
\begin{equation}\label{eq:qty}\DP(\mathcal{S}, b) -  (p/q)\success(\mathcal{S}, b)   =       \frac{1}{dK^2} \left( \sum_{t \in [K]: b_t= 1} H_{t - 1}     -   \sum_{t \in [K]: b_t= 0} (p/q)H_{t - 1}   \right). 
\end{equation}
We regard the quantity $H_{t - 1}$ as being composed of a contribution of 1 from each of the chips on the stack after processing $b_{t - 1}$.  We rewrite the right-hand side of Eq.~(\ref{eq:qty}) as a sum of the total contributions from each chip.  For a chip $c$, and for $z \in \{0, 1\}$, let 
\[n_{c, z} := \left|\{t \in  [K]: b_t = z, \text{ and }c \text{ is on the stack immediately after processing }b_{t - 1}\}\right|.    \] 
We then have
\begin{equation}\label{eq:qty2}\DP(\mathcal{S}, b) -  (p/q)\success(\mathcal{S}, b)   =   \frac{1}{dK^2} \sum_{c} (n_{c, 1} - (p/q)n_{c, 0}).
\end{equation}

Fix attention to some chip $c$, which was placed on the stack while processing the $i_c$-th bit, for some $i_c \in [K - 1]$.  
First assume that $c$ was later removed from the stack, and let $j_c \in [K - 1]$ be the index of the bit whose processing caused $c$ to be removed (thus, $b_{j_c} = 1$).
Then the stack was not empty after processing bits $i_c, \ldots, j_c - 1$, since in particular, the stack contained $c$.  Thus each 1 appearing in $(b_{i_c + 1}, \ldots b_{j_c - 1})$ caused exactly $q$ chips to be removed from the stack.  The removal caused by $[b_{j_c} = 1]$ removes some number $r_c \leq q$ of chips.
Also, each 0 appearing in the same range causes $p$ chips to be added.  Now $n_{c, 0}, n_{c, 1}$ count the number of 0s and 1s respectively among $(b_{i_c +1}, \ldots, b_{j_c})$.  Thus we have
\[  H_{j_c } - H_{i_c}  =  pn_{c, 0}   - q (n_{c, 1} - 1) - r_c \leq pn_{c, 0}   - q (n_{c, 1} - 1) ,   \]
or rearranging,
\begin{equation}\label{nc1_vs_nc0} n_{c, 1} - (p/q)n_{c, 0}      \leq  (H_{i_c} - H_{j_c})/q    + 1 .  \end{equation}
The chip $c$ is added to the stack with $p - 1$ other bits while processing bit $i_c$.  Later, $c$ is removed from the stack when processing bit $j_c$, along with at most $q - 1$ other chips.  Thus we have 
\[H_{i_c} - H_{j_c} \leq p + q - 1,\]
and combining this with Eq.~(\ref{nc1_vs_nc0}) gives
\begin{equation}\label{chipbound1}
n_{c, 1} - (p/q)n_{c, 0} \leq (p + q - 1)/q + 1 < p/q + 2.
\end{equation}

Next suppose $c$ was added after processing bit $i_c \in [K - 1]$, but never removed from the stack.  Then the stack was nonempty after processing bit $i_c$ and remained nonempty from then on, so each 1 in $b_{i_c + 1}, \ldots, b_{K-1}$
caused exactly $q$ chips to be removed.  By reasoning similar to the previous case, we get
\[    n_{c, 1} - (p/q)n_{c, 0}    =  (H_{i_c} - H_{K - 1})/q.\]
Now, $c$ was added along with $p - 1$ other chips after processing $b_{i_c}$, and $c$ remains on the stack after processing $b_{K - 1}$.  It follows that $H_{i_c} - H_{K - 1} \leq p - 1$, so
\begin{equation}\label{chipbound2}
n_{c, 1} - (p/q)n_{c, 0} \leq (p - 1)/q.     
\end{equation}
Plugging Eqs.~(\ref{chipbound1}) and~(\ref{chipbound2}) into Eq.~(\ref{eq:qty2}), we bound
\begin{align*}   \DP(\mathcal{S}, b) -  (p/q)\success(\mathcal{S}, b)  < \frac{1}{dK^2} \sum_{c}(p/q + 2)  &< \frac{p^2/q + 2p}{dK} \\
&\text{(since at most $p(K-1)$ chips are ever used)}\\
&= \frac{1}{K}\left(\frac{p}{q}\cdot \frac{p}{d} + \frac{2p}{d}  \right)\\
&= \frac{1}{K}\left( \frac{\delta}{1 -\delta}\cdot \delta + 2\delta \right)\\
&= O\left( \frac{\delta}{(1 - \delta)K}  \right).
\end{align*}
Since $(p/q) = \delta/(1 - \delta)$, this establishes condition (iii), and completes the proof of Lemma~\ref{lem:finite}.
\end{proof}


\section{More on Bit-Prediction}\label{automsec}

After thinking hard about car-streams and getting across the road safely, our frog had developed a taste for prediction.  In this section we present an extension of Puzzle~\ref{puz3} that is able to predict single bits from significantly more general classes of binary sequences.

\subsection{Bit-prediction algorithms}\label{bitpred_defsec}

Our result concerns the setting in which an observer is asked to correctly predict a single bit of their choice from a sequence $x$.  Unlike the frog crossing problem, in which the frog needed to correctly predict a 0, in this problem the algorithm is allowed to predict either a 0 or a 1.  Thus we need to modify our definition of frog-strategies (in the obvious way), as follows.
A \emph{bit-prediction strategy} is a collection 
\[ \mathcal{S} = \{\pi_{\cS, b} : b \in \seq  \},\]
where each $\pi_{\cS, b}$ is now a probability distribution over $(\bN \times \{0, 1\}) \cup \{\infty\}$.  We require that for all $b = (b_1, b_2, \ldots), b' = (b'_1, b'_2, \ldots)$, and all $i \in \bN, z \in \{0, 1\}$,
\[(b_1, \ldots, b_{i-1} ) =  (b'_1, \ldots, b'_{i-1} )   \Rightarrow  \pi_{\cS, b}((i, z)) = \pi_{\cS, b'}((i, z)). \]
That is, $\pi_{\cS, b}((i, z))$ depends only on $b_1, \ldots, b_{i - 1}$.  As in the frog-crossing setting, our bit-prediction strategies will be defined so that this constraint clearly holds.

Define the success probability
\[   \success^{\text{bit-pred}}(\mathcal{S}, b)  :=  \sum_{i \in \bN} \pi_{\cS, b}((i, b_i))         \]
as the probability that $\mathcal{S}$ correctly predicts a bit of $b$.  For a subset $A \subseteq \seq$, define $\success^{\text{bit-pred}}(\mathcal{S}, A) := \inf_{b \in A} \success^{\text{bit-pred}}(\mathcal{S}, b)$.

\subsection{Finite automata}

To state our result, we need the familiar notion of a \emph{finite automaton} over a binary alphabet.  Formally, this is a 3-tuple $M = (Q, s, \Delta)$, where: 
\begin{itemize}
\item $Q$ is a finite set of states; 
\item $s \in Q$ is the designated \emph{starting state}; 
\item $\Delta: Q \times \{0, 1\} \rightarrow Q$ is the \emph{transition function}.  
\end{itemize}
For $q \in Q$, $B \subseteq Q$, say that $B$ is \emph{accessible from $q$} if there exists a sequence $y_1, \ldots, y_m$ of bits and a sequence $q_0 = q, q_1, \ldots, q_m$ of states, such that
\begin{enumerate}
\item $\Delta(q_i, y_{i+1}) = q_{i+1}$ for $i = 0, 1, \ldots, m - 1$;
\item $q_m \in B$.
\end{enumerate}
Say that $B$ is \emph{strongly accessible} if, for any state $q$ that is accessible from the starting state $s$, $B$ is accessible from $q$.

Finite automata operate on infinite sequences $x \in \{0, 1\}^\omega$ as follows: we let $q_0(x) := s$, and inductively for $t \geq 1$ we define
\[ q_{t}(x) := \Delta (q_{t - 1}(x), x_{t}).\]
We say that $q_t(x)$ is the \emph{state of $M$ after $t$ steps} on the sequence $x$.

For a state $q \in Q$ we define $V_q(x)$, the \emph{visits to $q$ on $x$}, as
\[   V_q(x) := \{ t \geq 0: q_t(x) = q  \}.  \]
Similarly, for $B \subseteq Q$, define $V_B(x)$ as $V_B(x) := \{ t \geq 0: q_t(x) \in B  \}$.

\subsection{Statement of the result}\label{pred_statement}

Say we are presented with the bits of some unknown $x \in \seq$ sequentially.  We assume that $x$ is ``nice'' in the following sense: for some known finite automaton $M$, there is a set $B \subseteq Q$ of ``bad'' states of $M$, which we assume $M$ visits only infrequently when $M$ is run on $x$.  We show that, if $B$ is strongly accessible, we can successfully predict a bit of $x$ with high probability.

First recall the definition of weak sparsity from Section~\ref{puzsec}.  We say that a subset $S \subseteq \{0, 1, 2, \ldots\}$ is weakly sparse if its characteristic sequence is weakly sparse.  We prove:

\begin{theorem}\label{thm:fa_predict}  Let $M = (Q, s, \Delta)$ be a finite automaton, and let $B \subseteq Q$ be a strongly accessible set of states.  Define 
\[A_{B, ws} := \{x \in \seq: V_B(x) \text{ is weakly sparse}  \}. \]
Then for all $\eps > 0$, there exists a bit-prediction strategy $\mathcal{S} = \mathcal{S}_{\eps}$ such that
\[ \success^{\text{\rm{bit-pred}}}(\mathcal{S}, A_{B, ws}) > 1 - \eps. \]
\end{theorem}

We make a few remarks before proving Theorem~\ref{thm:fa_predict}.  
First, simple examples show that the conclusion of Theorem~\ref{thm:fa_predict} can hold even in some cases where $B$ is not strongly accessible.  Finding necessary and sufficient conditions on $B$ could be an interesting question for future study.

Second, it is natural to ask whether a more ``quantitative'' version of Theorem~\ref{thm:fa_predict} can be given.  Let $A_{B, \eps-ws}$ be the set of sequences $x$ for which the characteristic sequence of $V_B(x)$ is $\eps$-weakly sparse (as defined in Section~\ref{prelimsec}).  If $B$ is strongly accessible then, by a slight modification of our proof of Theorem~\ref{thm:fa_predict}, one can derive a bit-prediction strategy $\mathcal{S}$ such that
\[   \success^{\text{bit-pred}}(\mathcal{S}, A_{B, \eps-ws}) > 1 - O\left(\ell \eps^{1/\ell}\right) ,  \]
where $\ell = |Q|$ is the number of states of the automaton $M$.  

Something like this weak form of dependence on $\eps$ is essentially necessary, as can be seen from the following example.  Let $M$ be an automaton with states $Q = \{1, 2, \ldots, \ell\}$, and define 
\[ \Delta(i, 1) := \min\{i+1, \ell\}, \quad{} \Delta(i, 0) := 1. \]
Let $B := \{\ell\}$, and consider running $M$ on a sequence $\mathbf{b}$ of independent unbiased bits.  Then with probability 1, $V_B(\mathbf{b})$ is $2^{-\ell + 1}$-weakly sparse.  On the other hand, no algorithm can predict a bit of $\mathbf{b}$ with success probability greater than $1/2$.

\subsection{Proof of Theorem~\ref{thm:fa_predict}}

Let $A_{ws} \subseteq \seq$ denote the set of weakly sparse sequences.
Given a sequence $x = (x_1, x_2, \ldots)$, define $\neg x := (\neg x_1, \neg x_2, \ldots)$.  Say that $x$ is \emph{co-weakly sparse}, and write $x \in A_{co-ws}$, if $\neg x \in A_{ws}$.
To prove Theorem~\ref{thm:fa_predict}, we need two lemmas.  The following lemma follows easily from Theorem~\ref{thm:main}:   

\begin{lemma}\label{lem:cows}  Given $\delta > 0$, there exists a bit-prediction strategy $\mathcal{P} = \mathcal{P}_{\delta}$ such that
\[   \success^{\text{\rm{bit-pred}}}(\mathcal{P}, A_{ws} \cup A_{co-ws}) > 1 - \delta.    \]
$\mathcal{P}$ also has the ``safety'' property that for any $x \in \seq$, the probability that $\mathcal{P}$ outputs an incorrect bit-prediction on $x$ is at most $\delta$.
\end{lemma}

\begin{proof}  First, note that a frog-strategy (as defined in Section~\ref{prelimsec}) can be regarded as a bit-prediction strategy that only ever predicts a 0.  
Let $\eps = \gamma := \delta/4$.  The bit-prediction strategy $\mathcal{P}$, given access to some sequence $b$, simulates the frog-strategy $\mathcal{S}_{\eps, \gamma}$ from Theorem~\ref{thm:main} on $b$, and simultaneously simulates an independent copy of $\mathcal{S}_{\eps, \gamma}$ on $\neg b$.  If $\mathcal{S}_{\eps, \gamma}(b)$ ever outputs a prediction (i.e., that the next bit of $b$ will be 0), $\mathcal{P}$ immediately outputs the same prediction.  On the other hand, if $\mathcal{S}_{\eps, \gamma}(\neg b)$ ever outputs a prediction (that the next bit of $\neg b$ will be 0), then $\mathcal{P}$ predicts that the next bit of $b$ will be 1.  If both simulations output predictions simultaneously, $\mathcal{P}$ makes an arbitrary prediction for the next bit.

To analyze $\mathcal{P}$, say we are given input sequence $b \in A_{ws} \cup A_{co-ws}$.  First suppose $b \in A_{ws}$.  Then $\mathcal{S}_{\eps, \gamma}(b)$ outputs a correct prediction with probability $> 1 - \eps - \gamma$.  Also, by the safety property of $\mathcal{S}_{\eps, \gamma}$ shown in Theorem~\ref{thm:main}, the probability that $\mathcal{S}_{\eps, \gamma}(\neg b)$ outputs an incorrect prediction about $\neg b$ is at most $\eps + \gamma$.  Thus the probability that $\mathcal{P}$ outputs a correct prediction on $b$ is greater than $1 - 2\eps - 2\gamma = 1 - \delta$.

The case where $b \in A_{co-ws}$ is analyzed similarly.  Finally, the safety property of $\mathcal{P}$ follows from the safety property of $\mathcal{S}_{\eps, \gamma}$. \end{proof}

For the next lemma, we need some further definitions.  Fix a finite automaton $M = (Q, s, \Delta)$.  For $x \in \seq$, let 
\[Q_{\text{inf}}(x)  := \{ q \in Q: |V_q(x)| = \infty   \}. \]
Of course, $Q_{\text{inf}}(x)$ is nonempty  since $Q$ is finite.
If $q \in Q_{\text{inf}}(x)$, define a sequence $x^{(q)} \in \seq$ as follows.  If $V_q(x) =  \{t(1), t(2), \ldots, \}$ where $0 \leq t(1) < t(2) < \ldots$, we define
\[  x^{(q)}_i :=  x_{t(i) + 1} .       \]
In words: if $M$ is run on $x$, the $i$-th bit of $x^{(q)}$ records the bit of $x$ seen immediately after the $i$-th visit to state $q$.  If $q \notin Q_{\text{inf}}(x)$, we define $x^{(q)} \in \{0,1\}^\ast$ similarly; in this case, $x^{(q)}_i$ is undefined if $M$ visits state $q$ fewer than $i$ times while running on $x$.

The following lemma gives us a useful property obeyed by sequences $x$ from the set $A_{B, ws}$ (defined in the statement of Theorem~\ref{thm:fa_predict}).

\begin{lemma}\label{lem:structure}  Given $M = (Q, s, \Delta)$, suppose $B \subseteq Q$ is strongly accessible.  If $x \in A_{B, ws}$, then there exists a state $q \in Q_{\text{\rm{inf}}}(x)$ such that
\[    x^{(q)}  \in  A_{ws} \cup A_{co-ws}.    \]
\end{lemma}

\begin{proof}  We prove the contrapositive.  Assume that all $q \in Q_{\text{inf}}(x)$ satisfy $x^{(q)} \notin   A_{ws} \cup A_{co-ws}$; we will show that $x \notin A_{B, ws}$.

Say that a state $q \in Q$ is \emph{frequent (on $x$)} if there exist $\alpha, \beta > 0$ such that for all $T \in \mathbb{N}$,
\[  \left|  V_q(x)  \cap \{0, 1, \ldots, T-1\}   \right|   \geq \alpha T - \beta .        \]
Let $F$ denote the set of frequent states.  Clearly $F \subseteq Q_{\text{inf}}(x)$.  We will show:
\begin{enumerate}
\item $F = Q_{\text{inf}}(x)$;
\item $F$ contains a state from $B$.
\end{enumerate}
Item 2 will immediately imply that $x \notin A_{B, ws}$, as desired.

For each $q \in Q_{\text{inf}}(x)$, our assumption $x^{(q)} \notin   A_{ws} \cup A_{co-ws}$ implies that there is a $\delta_q  \in (0, 1/2)$ and a $K_{q} > 0$ such that for $k \geq K_q$, 
\begin{equation}\label{balance}\delta_q < \frac{1}{k}\left( x^{(q)}_1 + \ldots +  x^{(q)}_k   \right)    < 1 - \delta_q.   \end{equation}
Let $\delta := \min \delta_q$.  Choose a value $T^\star > 0$ such that each $q \in Q_{\text{inf}}(x)$ appears at least $K_q$ times among $(q_0(x), q_1(x), \ldots, q_{T^\star - 1}(x))$.
Choose a second value $R > 0$, such that any $q \notin Q_{\text{inf}}(x)$ occurs fewer than $R$ times in the infinite sequence $(q_0(x), q_1(x), \ldots)$.

Let $\ell = |Q|$.  Fix any $t \in \mathbb{N}$ satisfying 
\[t \geq \max\left\{ \frac{\ell  R}{  \delta^{2(\ell - 1)}} , T^\star \right\}.\]

  By simple counting, some $q^\star \in Q$ occurs at least $t/\ell$ times in $(q_0(x), q_1(x), \ldots, q_{t - 1}(x))$.  We have $t/\ell > R$, so this $q^\star$ must lie in $Q_{\text{inf}}(x)$.
Eq.~(\ref{balance}) then implies that the states $\Delta(q^\star, 0), \Delta(q^\star, 1)$ each appear at least $\delta t/\ell - 1 > \delta^2 t / \ell$ times among $(q_0(x), q_1(x), \ldots, q_{t-1}(x))$.  Now $\delta^2 t / \ell > R$, so we have $\Delta(q^\star, 0), \Delta(q^\star, 1) \in Q_{\text{inf}}(x)$.  

Iterating this argument $(\ell - 1)$ times, we conclude that every state $q$ reachable from $q^\star$ by a sequence of $(\ell - 1)$ or fewer transitions lies in $Q_{\text{inf}}(x)$, and appears at least $\delta^{2(\ell - 1)}t/\ell = \Omega(t)$ times among $(q_0(x), q_1(x), \ldots, q_{t-1}(x))$.  But \emph{every} $q \in Q_{\text{inf}}(x)$ is reachable from $q^{\star}$ by at most $(\ell - 1)$ transitions.  Thus $F =  Q_{\text{inf}}(x)$, proving Item 1 above.

The argument above shows that if $q \in Q_{\text{inf}}(x)$, then $\Delta(q, 0), \Delta(q, 1) \in Q_{\text{inf}}(x)$ as well.  Recall that $B$ is strongly accessible; it follows that $Q_{\text{inf}}(x) \cap B$ is nonempty, proving Item 2 above.  This proves Lemma~\ref{lem:structure}.
\end{proof}

We can now complete the proof of Theorem~\ref{thm:fa_predict}.  Let $Q = \{p_1, \ldots, p_\ell\}$, where $\ell = |Q|$.  We may assume $\ell > 1$, for otherwise $A_{B, ws} = \emptyset$ and there is nothing to show.
  Given $\eps > 0$, let $\delta := \eps /(2\ell)$.  
We define the algorithm $\mathcal{S} = \mathcal{S}_{\eps}$ as follows.  $\mathcal{S}$ runs in parallel $\ell$ different simulations 
\[\mathcal{P}[1], \ldots, \mathcal{P}[\ell]\]
of the algorithm $\mathcal{P}_{\delta}$ from Lemma~\ref{lem:cows}.  $\mathcal{P}[j]$ is run, not on the input sequence $x$ itself, but on the subsequence $x^{(p_j)}$.  To determine which simulation receives each successive bit of $x$, the algorithm $\mathcal{S}$ simply simulates $M$ on the bits of $x$ seen so far.  (Note that, if $p_j \notin Q_{\text{inf}}(x)$, then the simulation $\mathcal{P}[j]$ may ``stall'' indefinitely without receiving any further input bits.)  

Suppose that the simulation $\mathcal{P}[j]$ outputs a prediction $z \in \{0, 1\}$ after seeing the $i$-th bit of $x^{(p_j)}$, and that we subsequently reach a time $t$ such that $q_t(x) = p_j$ is the $(i+1)$-st visit to state $p_j$.  The algorithm $\mathcal{S}$ then predicts that $x_{t+1} = x^{(p_j)}_{i+1} = z$.   

We now analyze $\mathcal{S}$.  Fix any $x \in A_{B-ws}$.  By the safety property of Lemma~\ref{lem:cows}, each $\mathcal{P}[j]$ outputs an incorrect prediction with probability at most $\delta$, so the overall probability of an incorrect prediction is at most $\ell \delta  = \eps/2$.  Also, since $x \in A_{B, ws}$, Lemma~\ref{lem:structure} tells us that there exists a $p_j \in Q_{\text{inf}}(x)$ such that $x^{(p_j)} \in A_{ws} \cup A_{co-ws}$.  Thus, if $\mathcal{P}[j]$ is run individually on $x^{(p_j)}$, $\mathcal{P}[j]$ outputs a correct prediction with probability greater than $1 - \delta$.  We conclude that
\[ \success^{\text{bit-pred}}\left(  \mathcal{S} , x  \right)  > (1 - \delta) - \eps/2 > 1 - \eps, \]
using $\ell > 1$.  This proves Theorem~\ref{thm:fa_predict}.

\section{The Density Prediction Game}\label{forecastsec}

In this section we prove Theorem~\ref{thm:denspred} from Section~\ref{sec:further}.  The proof uses a technique from the analysis of martingales that seems to be folklore; my understanding of this technique benefited greatly from conversations with Russell Impagliazzo.

For any fixed $\delta, \eps$, our prediction strategy will work entirely within a finite interval $(x_1, \ldots, x_T)$ of the sequence $x$.
We note that, to derive a $(\delta, \eps)$-successful strategy over this interval, it suffices to show that for every distribution $\mathcal{D}$ over $\{0, 1\}^T$, there exists a strategy $\mathcal{S}_{\mathcal{D}}$ that is $(\delta, \eps)$-successful when played against $\mathcal{D}$.  This follows from the minimax theorem of game theory, or from the result of Sandroni~\cite{sand03} mentioned in Section~\ref{relatedsec}.  However, this observation would lead to a nonconstructive proof of Theorem~\ref{thm:denspred}, and in any case does not seem to make the proof any simpler.  Thus we will not follow this approach.

Let $\delta, \eps > 0$ be given; we give a forecasting strategy $\mathcal{S} = \mathcal{S}_{\delta, \eps}$ for the density prediction game, and prove that $\mathcal{S}$ is $(\delta, \eps)$-successful.
Set $n := \lceil 4/(\delta \eps^2) \rceil$.  Our strategy will always make a prediction about an interval $x_{a}, \ldots, x_{b}$ where $a \leq b \leq 2^n$. The strategy $\mathcal{S}$ is defined as follows:

\begin{enumerate}
\item Choose $R \in \{1, \ldots, n\}$ uniformly.  Choose $S$ uniformly from $\{1, \ldots , 2^{n - R}\}$.
\item Ignore the first $t = (S - 1) \cdot 2^R$ bits of $x$.  Observe bits $x_{t + 1}, \ldots, x_{t + 2^{R - 1}}$, and let $p$ be the fraction of 1s in this interval.  Immediately after seeing $x_{t + 2^{R - 1}}$, predict:
\[ \text{\textit{``Out of the next $2^{R - 1}$ bits, a $p$ fraction will be 1s.''}}  \]
\end{enumerate}

We now analyze $\cS$.  To do so, it is helpful to describe $\cS$ in a slightly different fashion.  Let us re-index the first $2^n$ bits of our sequence $x$, considering each such bit to be indexed by a string $z \in \{0, 1\}^n$.  We use lexicographic order, so that the sequence is indexed $x_{0^n}, x_{0^{n-1}1}, x_{0^{n - 2}1 0}$, and so on.

Let $T$ be a directed binary tree of height $n$, whose vertices at depth $i$ ($0 \leq i \leq n$) are indexed by binary strings of length $i$; in particular, the root vertex is labeled by the empty string.  If $i < n$ and $y \in \{0, 1\}^i$, the vertex $v_{y}$ has left and right children $v_{y0}, v_{y1}$ respectively.    Each leaf vertex is indexed by an $n$-bit string $z$, and any such vertex $v_z$ is labeled with the bit $x_z$.  

For $y \in \{0, 1\}^\ast$, let $T_y$ denote the subtree of $T$ rooted at $v_y$.  A direct translation of the strategy $\cS$ into our current perspective gives the following equivalent description of $\cS$:
\begin{itemize}
\item[1'.] Choose $R \in \{1, \ldots, n\}$ uniformly.  Starting at the root of $T$, take a directed, unbiased random walk of length $n - R$, reaching a vertex $v_Y$ where $Y \in \{0, 1\}^{n - R}$.

\item[2'.] Observe the bits of $x$ that label leaf vertices in $T_{Y0}$, and let $p$ be the fraction of 1s seen among these bits.  Immediately after seeing the last of these bits, predict:
\[ \text{\textit{``Out of the next $2^{R - 1}$ bits of $x$ (i.e., those labeling leaf vertices in $T_{Y1}$), a $p$ fraction will be 1s.''}}  \]
\end{itemize}

To analyze $\cS$ in this form, fix any binary sequence $x$.  We consider the random walk performed in $\cS$ to be extended to an unbiased random walk of length $n$.  The walk terminates at some leaf vertex $v_{Z}$, where $Z = (z_1, \ldots, z_n)$ is uniform over $\{0, 1\}^n$.

For $0 \leq i \leq n$ and $y \in \{0, 1\}^i$, define
\[\rho(y) :=  2^{i - n} \sum_{w \in \{0, 1\}^{n - i}}x_{yw}   \]
as the fraction of 1s among the labels of leaf vertices of $T_y$.
For $0 \leq t \leq n$, define the random variable 
\[X(t) := \rho(z_1, \ldots, z_t), \]
defined in terms of $Z$, where $X(0) = \rho(\emptyset)$.  The sequence $X(0), \ldots, X(n)$ is a martingale; we follow a folklore technique by analyzing the squared differences between terms in the sequence.
 First, we have $X(t) \in [0, 1]$, so that $(X(n) -  X(0))^2 \leq 1$.  On the other hand, 
\begin{align}\label{dropout}  
\bE[(X(n) -  X(0))^2]  &= \bE\left[ \left( \sum_{0 \leq t < n}(X(t+1) -  X(t))  \right)^2  \right]  \nonumber \\
&= \bE  \left[        \sum_{0 \leq t < n} (X(t + 1) - X(t))^2         \right]  + \bE\left[    2 \sum_{0 \leq s < t < n}(X(s + 1) - X(s))(X(t + 1) - X(t))   \right].
\end{align}
Now, for $0 \leq s < t < n$ and for any outcome of the bits $z_1, \ldots, z_{t}$ (which determine  $X(s)$, $X(s+1)$,\\ and $X(t)$), we have
\begin{align*}  
\bE[(X(t+1) - X(t) )| z_1, \ldots, z_{t}] &=      \bE_{z_{t+1} \in \{0, 1\}}[(\rho(z_1, \ldots, z_{t+1})] - \rho(z_1, \ldots, z_t)  \\
&=  \frac{1}{2}\left[ \rho(z_1, \ldots, z_t, 0) + \rho(z_1, \ldots, z_t, 1)  \right] -\rho(z_1, \ldots, z_t)  \\
&= 0.
\end{align*}
Thus the second right-hand term in Eq.~(\ref{dropout}) is 0, and
\begin{equation}\label{expbound} 
\bE[(X(n) -  X(0))^2]  =       \sum_{0 \leq t < n}\bE \left[( X(t + 1) - X(t))^2         \right]   . 
\end{equation}

Next we relate this to the accuracy of our guess $p$.
Let $p^\ast$ be the fraction of 1s in $T_{Y1}$, i.e., the quantity $\mathcal{S}$ attempts to predict; note that $p^\ast$ and $p$ are both random variables.  From the definitions, we have
\[ p = \rho( Y 0 ), \quad{} p^\ast = \rho (  Y1 ), \quad{} X(n - R) = \frac{1}{2}\left( p + p^\ast \right).    \]
Also, 
\[X(n - R + 1) =  \left\{
	\begin{array}{ll}
		p  & \mbox{if } z_{n - R + 1}  = 0, \\
		p^\ast & \mbox{if } z_{n - R + 1} =  1.
	\end{array}
\right. \]
Thus we have the identity
\begin{align*}
(X(n - R + 1) - X(n - R))^2 &= \frac{1}{4}  \left( p - p^\ast \right)^2.
\end{align*}
Now, $n - R$ is is uniform over $\{0, 1, \ldots, n - 1\}$, and independent of $Z$.  It follows from Eq.~(\ref{expbound}) that
\[\bE[(X(n - R + 1) - X(n - R))^2] = \frac{1}{n}\bE[(X(n) -  X(0))^2] \leq 1/n .\]
Combining, we have
\begin{equation}\label{upbound} \bE [  \left( p - p^\ast \right)^2 ]   \leq  4/n.   
\end{equation}
On the other hand,
\begin{equation}\label{lowbound} \bE [  \left( p - p^\ast \right)^2 ]   \geq  \Pr[|p - p^\ast| \geq \eps] \cdot \eps^2.   
\end{equation}
Combining Eqs.~(\ref{upbound}) and~(\ref{lowbound}), we obtain
\[ \Pr[ |p - p^\ast|   \geq \eps]  \leq  4/(n\eps^2) \leq \delta,    \]
by our setting $n = \lceil 4/(\delta \eps^2) \rceil$.
This proves Theorem~\ref{thm:denspred}.

\section{Questions for Future Work}
\begin{enumerate}
\item Fix some $p \in [1/2,1]$; is there a satisfying characterization of the sets $A \subseteq \{0, 1\}^\omega$ for which some bit-prediction strategy (as defined in Section~\ref{bitpred_defsec}) succeeds with probability $\geq p$ against all $x \in A$?  Perhaps there is a  characterization in terms of some appropriate notion of dimension, analogous to the gale characterizations of Hausdorff dimension~\cite{lutz03} and packing dimension~\cite{ath07}.

\item Could the study of computationally bounded bit-prediction strategies be of value to the study of complexity classes, by analogy to the study of computationally bounded gales in~\cite{lutz03, ath07} and in related work?

\item Find necessary and sufficient conditions on the set $B$ of ``infrequently visited'' states, for the conclusion of Theorem~\ref{thm:fa_predict} (in Section~\ref{pred_statement}) to hold.

\item Our $(\delta, \eps)$-successful forecasting strategy in Section~\ref{forecastsec} always makes a forecast about an interval of bits within $x_1, \ldots, x_{m}$, where $m= 2^{O(\delta^{-1} \eps^{-2})}$.  It would be interesting to know whether some alternative strategy could make forecasts within a much smaller interval---for instance, with $m = \poly(\delta^{-1}, \eps^{-1})$.  It would also be interesting to look at a setting in which the forecaster is allowed to make predictions about sets other than intervals.
\end{enumerate}

\section{Acknowledgements}

I thank John Hitchcock for pointing me to the ignorant-forecasting literature, and Russell Impagliazzo for helpful discussions.  I also thank the many people who tried the puzzles in this paper.

\bibliographystyle{halpha}

\end{document}